%% file: Square_roots-maxdeg-arxiv.tex
\documentclass{llncs}

\usepackage{tikz}
\usepackage{tkz-graph}

\usepackage{amsmath,graphics,color,comment}
\usepackage{graphicx}

\newcommand{\NP}{{\sf NP}}

\newcommand{\dist}{{\rm dist}}
\newcommand{\diam}{{\rm diam}}
\newcommand{\tw}{{\mathbf{tw}}}
\newcommand{\pw}{{\mathbf{pw}}}

\begin{document}

\title{Squares of Low Maximum Degree
\thanks{This paper received support from  EPSRC (EP/G043434/1), ERC (267959) and ANR project AGAPE. 
The results of this paper appeared (with alternative proofs) as an extended abstract in the proceedings of WG 2013~\cite{CochefertCGKP13}.
}}
\author{
Manfred Cochefert\inst{1}
\and
Jean-Fran\c{c}ois Couturier\inst{2}
\and
Petr A. Golovach\inst{3}
\and
Dieter Kratsch\inst{1}
\and
Dani{\"e}l Paulusma\inst{4}
\and
Anthony Stewart\inst{4}
}

\institute{
Laboratoire d'Informatique Th\'eorique et Appliqu\'ee,
Universit\'e de Lorraine,\\ 57045 Metz Cedex 01, France,\\  \texttt{manfred.cochefert@gmail.com},  \texttt{dieter.kratsch@univ-lorraine.fr}
\and
IFTS/URCA, CReSTIC, 7 Boulevard Jean Delautre, 08005 Charleville-M\'ezi\`eres, France,
\texttt{jean-francois.couturier@univ-reims.fr}
\and
Department of Informatics, 
University of Bergen, PB 7803, 5020 Bergen, Norway,\\
\texttt{petr.golovach@ii.uib.no}
\and
School of Engineering and  Computing Sciences, Durham University,\\ Durham DH1 3LE, UK,
\texttt{\{daniel.paulusma,a.g.stewart\}@durham.ac.uk}
}
\maketitle

\begin{abstract}A graph $H$ is a square root of a graph $G$ if $G$ can be obtained from $H$ by adding an edge between any two vertices in $H$ that are of distance~2.
The {\sc Square Root} problem is that of deciding whether a given graph admits a square root.
This problem is only known to be \NP-complete for chordal graphs and polynomial-time solvable for non-trivial minor-closed
graph classes and a very limited number of other
graph classes. 
We prove that {\sc Square Root} is $O(n)$-time solvable for graphs of maximum degree~5 and $O(n^4)$-time solvable
for graphs  of maximum degree at most~6. 
\end{abstract}

\section{Introduction}\label{s-intro}

The {\it square} $H^2$ of a graph $H=(V_H,E_H)$ is the graph with vertex set~$V_H$, such that any two distinct vertices 
$u,v\in V_H$ are adjacent in $H^2$ if and only if $u$ and $v$ are of distance at most~2 in $H$.
 In this paper we study the reverse concept: a graph $H$ is a {\it square root} of a graph~$G$ if $G=H^2$. There exist  graphs with no square root (such as graphs with a cut-vertex), graphs with a unique square root (such as squares of cycles of length at least~7) as well as graphs with more than one square root (such as complete graphs). 

In 1967 Mukhopadhyay~\cite{Mukhopadhyay67} characterised the class of connected graphs with a square root.
However, in 1994, Motwani and Sudan~\cite{MotwaniS94} showed that  the decision problem {\sc Square Root}, which asks whether a given graph admits 
a square root, is \NP-complete. 
As such, it is natural to restrict the input to special graph classes in order to obtain polynomial-time results. For several well-known graph classes the complexity of {\sc Square Root} is still unknown. For example, 
Milanic and Schaudt~\cite{MilanicS13} posed the complexity of {\sc Square Root} restricted to split graphs and cographs as open problems.
In Table~\ref{t-survey} we survey the known results. 

The last two rows in Table~\ref{t-survey} correspond to the results in this paper. More specifically, 
we prove in Section~\ref{s-5} that {\sc Square Root} is linear-time solvable for graphs of maximum degree at most~5 via a reduction to graphs of bounded treewidth and in Section~\ref{s-6} that {\sc Square Root} is $O(n^4)$-time solvable for graphs of maximum degree at most~6 via a reduction to graphs of bounded size. 

\begin{table}
\begin{center}
 \begin{small}
\begin{tabular}{|c|c|c|}
 \hline
graph class ${\cal G}$&complexity\\
\hline
planar graphs~\cite{LinS95}&linear\\
\hline
non-trivial and minor-closed~\cite{NT14} &linear\\
\hline
$K_4$-free graphs~\cite{GKPS}&linear\\
\hline
$(K_r,P_t)$-free graphs~\cite{GKPS} &linear\\
\hline
$3$-degenerate graphs~\cite{GKPS}&linear\\
\hline
graphs of maximum degree $\leq 5$ &linear\\
\hline
graphs of maximum degree $\leq 6$ &polynomial\\
\hline
graphs of maximum average degree $<\frac{46}{11}$~\cite{GKPS16} &polynomial\\
\hline
line graphs~\cite{MOS14}&polynomial\\
\hline
trivially perfect graphs~\cite{MilanicS13}&polynomial\\
\hline
threshold graphs~\cite{MilanicS13}&polynomial\\
\hline
chordal graphs~\cite{LauC04}&\NP-complete\\
\hline
\end{tabular}
\end{small}
\end{center}
\caption{The known results for {\sc Square Root} restricted to some special graph class~${\cal G}$. Note that the row for planar graphs is absorbed by the row below. The two unreferenced results are the results of this paper.}
\label{t-survey}
\end{table}
\vspace*{-1cm}

Also the {\sc ${\cal H}$-Square Root} problem, which is that of testing whether  a given graph has a square root that belongs to some specified graph class ${\cal H}$, has also been well studied. We refer to Table~\ref{t-survey2} for a survey of the known results on {\sc ${\cal H}$-Square Root}.

\begin{table}
\begin{center}
 \begin{small}
\begin{tabular}{|c|c|c|}
 \hline
graph class ${\cal H}$ &complexity\\
\hline
trees~\cite{LinS95}&polynomial\\
\hline
proper interval graphs~\cite{LauC04}&polynomial\\
\hline
bipartite graphs~\cite{Lau06}&polynomial\\
\hline
block graphs~\cite{LeT10}&polynomial\\
\hline
strongly chordal split graphs~\cite{LeT11}&polynomial\\
\hline
ptolemaic graphs~\cite{LOS15}&polynomial\\
\hline
3-sun-free split graphs~\cite{LOS15}&polynomial\\
\hline
cactus graphs~\cite{GKPS16b}&polynomial\\
\hline
graphs with girth at least $g$ for any fixed $g\geq 6$~\cite{FarzadLLT12}&polynomial\\
\hline
graphs of girth at least~5~\cite{FarzadK12} &\NP-complete\\
\hline
graphs of girth at least~4~\cite{FarzadLLT12} &\NP-complete\\
\hline
split graphs~\cite{LauC04} &\NP-complete\\
\hline
chordal graphs~\cite{LauC04} &\NP-complete\\
\hline
\end{tabular}
\end{small}
\end{center}
\caption{The known results for {\sc ${\cal H}$-Square Root} restricted to various graph classes~${\cal H}$. The result for 3-sun-free split graphs has been extended to a number of other subclasses of split graphs in~\cite{LOS}.}
\label{t-survey2}
\end{table}
\vspace*{-0.75cm}

Finally both {\sc Square Root} and  {\sc ${\cal H}$-Square Root} have been studied under the framework of parameterized complexity.
The generalization of {\sc Square Root} that takes as input a graph $G$ with two subsets $R$ and $B$ of edges that need to be included or excluded, respectively, in any solution (square root)\footnote{We give a formal definition of this generalization in Section~\ref{s-6}, as we need it for proving that {\sc Square Root} is $O(n^4)$-time solvable for graphs of maximum degree at most~6.} has a kernel of size $O(k)$ for graphs that can be made planar after removing at most $k$ vertices~\cite{GKPS16}. The problems of testing whether a connected $n$-vertex graph with $m$ edges has a square root with at most $n-1+k$ edges and
whether such a graph has a square root with at least $m-k$ edges are both fixed-parameter tractable when parameterized by~$k$~\cite{CCGKP}.

\section{Preliminaries}\label{sec:defs}
We only consider finite undirected graphs without loops or multiple edges. 
We refer to the textbook by Diestel~\cite{Diestel10} for any undefined graph terminology.

Let $G$ be a graph.
We denote the vertex set of~$G$ by $V_G$ and the edge set by $E_G$. 
The {\it length} of a path or a cycle is the number of edges of the path or cycle, respectively.
The \emph{distance} $\dist_G(u,v)$ between a pair of vertices $u$ and $v$ of~$G$ is the number of edges of a shortest path between them. 
The diameter $\diam(G)$ of~$G$ is the maximum distance between two vertices of $G$. 
The \emph{neighbourhood} of a vertex $u\in V_G$ is defined as $N_G(u) = \{v\; |\; uv\in E_G\}$.
The {\it degree} of a vertex $u\in V_G$ is defined as $d_G(u)=|N_G(u)|$.
The {\it maximum degree} of $G$ is $\Delta(G)=\max\{d_G(v)\; |\; v\in V_G\}$.
A vertex of degree~1 and the (unique) edge incident to it are said to be a \emph{pendant} vertex and \emph{pendant} edge of $G$ respectively. 
A vertex subset of $G$ that consists of mutually adjacent vertices is called a {\it clique}. 

A \emph{tree decomposition} of a graph $G$ is a pair $(T,X)$ where $T$
is a tree and $X=\{X_{i} \mid i\in V_T\}$ is a collection of subsets (called {\em bags})
of $V_G$ such that the following three conditions hold: 
\begin{itemize}
\item[i)] $\bigcup_{i \in V_T} X_{i} = V_G$, 
\item[ii)] for each edge $xy \in E_G$, $x,y\in X_i$ for some  $i\in V_T$, and 
\item[iii)] for each $x\in V_G$ the set $\{ i \mid x \in X_{i} \}$ induces a connected subtree of $T$.
\end{itemize}
The \emph{width} of a tree decomposition $(\{ X_{i} \mid i \in V_T \},T)$ is $\max_{i \in V_T}\,\{|X_{i}| - 1\}$. The \emph{treewidth} $\tw(G)$ of a graph $G$ is the minimum width over all tree decompositions of $G$. If $T$ restricted to be a path, then we say that $(X,T)$ is a \emph{path decomposition} of a graph $G$. 
The \emph{pathwidth} $\pw(G)$ of $G$  is the minimum width over all path decompositions of $G$.
A class of graphs~${\cal G}$ has {\it bounded} treewidth (pathwidth) if there exists a constant $p$ such that the treewidth (pathwidth) of every graph from ${\cal G}$ is at most~$p$.

\section{Graphs of Maximum Degree at Most~5}\label{s-5}

In this section we prove that {\sc Square Root} can be solved in linear time for graphs of maximum degree at most~5 by showing that squares of maximum degree at most~5 have bounded pathwidth. The latter enables us to use the following two lemmas. The first lemma can either be proven via 
a dynamic programming algorithm or by a non-constructive proof based on Courcelle's theorem~\cite{Courcelle92} (see~\cite{GKPS16} for details).
The second lemma is due to Bodlaender.

\begin{lemma}\label{lem:tw}
The {\sc Square Root} problem can be solved in  
time $O(f(k)n)$ for $n$-vertex graphs of treewidth (or pathwidth) at most $k$.
\end{lemma}

\begin{lemma}[\cite{Bodlaender96}]\label{l-bod}
For any fixed constant~$k$, it is possible to decide in linear time whether the treewidth (or pathwidth) of a graph is at most $k$.
\end{lemma}

We now prove the key result of this section.

\begin{lemma}\label{lem:deg-5}
If $G$ is a graph with  $\Delta(G)\leq 5$ that has a square root, then $\pw(G)\leq 27$. 
\end{lemma}

\begin{proof}
Without loss of generality we assume that $G$ is connected; otherwise, we can consider the components of $G$ separately. Let $H$ be a square root of $G$.
Let $u\in V_G$. In $H$ we apply a breadth-first search (BFS) starting at $u$. This yields the
levels $L_0,\ldots,L_s$ for some $s\geq 0$, where $L_i(u)=\{v\in V_G\; |\; \dist_G(u,v)=i\}$ for $i=0,\ldots,s$.
Note that $L_0=\{u\}$ and that $L_0(u)\cup \cdots \cup L_s$ is a partition of $V_H=V_G$.
Let $T$ be the corresponding BFS-tree of $H$ rooted in $u$. Note that $T$ also defines a parent-child relation on the vertices of $G$. 

We prove the following two claims.

\medskip
\noindent
{\bf Claim A. }{\it Let $i\ge 2$. Then 
$x\in L_i$ implies that 
\begin{itemize}
\item[i)] $x$ has at most three children in $T$, and  
\item[ii)] for any $j\in\{i+1,\ldots,s-1\}$, $x$ has at most four descendants in $L_j\cup L_{j+1}$.
\end{itemize}
}

\begin{figure}[ht]
\centering\scalebox{0.95}{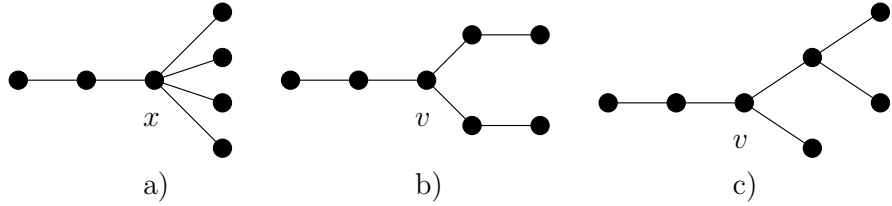}
\caption{Forbidden subgraphs for square roots of graphs of maximum degree at most~5.
\label{fig:deg-5}}
\end{figure}

\noindent
We prove Claim~A as follows.
First we show i) by observing that if $x$ had at least four children in $T$, then $H$ contains the subgraph shown in Fig.~\ref{fig:deg-5} a) and, therefore, $d_G(x)\geq 6$, which is a contradiction.

We now prove ii).
First assume that $x$ has exactly one descendant $y\in L_j$. Then $y$ has at most three children in $L_{j+1}$ due to i), and hence the total number of descendants of $x$ in $L_j\cup L_{j+1}$ is at most~4.
Now assume that $x$ has at least two descendants in $L_j$, say $y_1$ and $y_2$. 
Let $v$ denote the lowest common ancestor of $y_1,y_2$ in $T$. Note that $v$ is a descendant of $x$ or $v=x$, so $v\in L_k$ for some 
$i\leq k\leq j-1$.
Suppose $k<j-1$. Then  $H$ contains the subgraph shown in Fig.~\ref{fig:deg-5} b) and hence $d_G(v)\geq 6$, a contradiction.  Hence, $v\in L_{j-1}$ and, moreover, we may assume without loss of generality that $v$ is the parent in $T$ of all the descendants of $x$ in $L_j$ 
(as otherwise there exists a vertex $v'\in L_{j-1}$ with a neighbour $y_3\in L_j$, which means that the lowest common ancestor of $v$ and $v'$ has six neighbours in $G$). 

By i), we find that $v$ has at most three children. Hence, $x$ has at most three descendants in $L_j$. To obtain a contradiction, 
assume that $x$ has at least two descendants  $z_1,z_2$ in $L_{j+1}$. If $z_1,z_2$ have distinct parents we again find that  $H$ contains the forbidden subgraph shown in Fig.~\ref{fig:deg-5} b). 
If $z_1,z_2$ have the same parent, $H$ contains the subgraph shown in Fig.~\ref{fig:deg-5} c) and  hence $d_G(v)\geq 6$, a contradiction.
Hence, $x$ has at most one descendant in $L_{j+1}$. We conclude that  the total number of descendants of $x$ in $L_j\cup L_{j+1}$ is at most 4. 
This completes the proof of ii). Consequently we have proven Claim~A.

\medskip
\noindent
{\bf Claim B. }{\it 
$\pw(G)\leq \max\{ |L_i\cup L_{i+1}\cup L_{i+2}|\; |\;0\leq i\leq s\}-1.$}

\medskip
\noindent
We prove Claim B as follows.
Let $P$ be the path on vertices 
$0,\ldots,s$ (in the order of the path). We set $X_i=L_i\cup L_{i+1}\cup L_{i+2}$ for all $i\in \{0,\ldots,s\}$ and define $X=\{X_1,\ldots,X_s\}$.
We claim that $(X,P)$ is a path decomposition of $G$. This can be seen as follows.
Since the sets  $L_0,\ldots,L_s$ form a partition of $V_G$, we find that $\bigcup_{i=0}^{s} X_{i} = V_G$.
Moreover, for every edge $xy\in E_{G}$ with $x\in L_i$ and $y\in L_j$ we see that
$|i-j|\leq 2$. Hence, 
for each edge $xy \in E_{G}$, $x,y\in X_i$ for some  $i\in \{0,\ldots,s\}$. 
Finally, if $x\in X_i\cap X_j$ such that $i+1<j$, then $i+2=j$ and $x\in L_{i+2}\subseteq X_{i+1}$. Therefore, 
the set $\{ i \mid x \in X_{i} \}$ induces a subpath of $P$ for each $x\in V_G$. It remains to observe that the width of $(X,P)$ is 
$\max\{ |L_i\cup L_{i+1}\cup L_{i+1}|\; |\; 0\leq i\leq s\}-1$.
This completes the proof of Claim~B.

\medskip
\noindent
Because $d_G(u)\leq 5$, $|L_1\cup L_2|\leq 5$ and thus $|L_2|\leq 4$ and
$|L_0\cup L_1\cup L_2|\leq 6$. By Claim~A, each vertex of $L_2$ has at most three children in $T$ and at most four descendants in $L_3\cup L_4$.
Hence, $|L_1\cup L_2\cup L_3|\leq 17$ and $|L_2\cup L_3\cup L_4|\leq 20$.
For $j\in\{3,\ldots,s\}$, each vertex of $L_2$ has at most four descendants in $L_j\cup L_{j+1}$ and also at most four descendants in $L_{j+1}\cup L_{j+2}$ by Claim~A ii). Therefore, each vertex of $L_2$ has at most seven descendants in $L_j\cup L_{j+1}\cup L_{j+2}$. As $|L_2|\leq 4$, this means that $|L_j\cup L_{j+1}\cup L_{j+2}|\leq 28$.
We conclude that $|L_i\cup L_{i+1}\cup L_{i+2}|\leq 28$ for all 
$i\in\{0,\dots,s\}$. Consequently, Claim~B implies that $\pw(G)\leq 27$.\qed
\end{proof}

We are now ready to prove the main theorem of this section.

\begin{theorem}\label{thm:deg-5}
 {\sc  Square Root} can be solved in time $O(n)$ for  $n$-vertex graphs of maximum degree at most~$5$.  
\end{theorem}

\begin{proof}
Let $G$ be a graph with $\Delta(G)\leq 5$. By Lemma~\ref{l-bod} we can check in $O(n)$ time whether $\pw(G)\leq 27$. If $\pw(G)>27$, then $G$ has no square root by Lemma~\ref{lem:deg-5}. Otherwise, we solve {\sc Square Root} in $O(n)$ time by using Lemma~\ref{lem:tw}.\qed
\end{proof}

\noindent
\noindent
{\bf Remark 1.} The above approach cannot be extended to graphs of maximum degree at most~6. In order to see this, 
take a wall (see Figure~\ref{f-walls}) and subdivide each edge, that is, replace each edge~$uw$ by a path $uvw$ where $v$ is a new vertex.
This gives us a graph~$H$, such that $H^2$ has degree at most~6. 
A wall of height~$h$ has treewidth $\Omega(h)$ (see, for example,~\cite{Diestel10}).   
As subdividing an edge and adding edges does not decrease the treewidth of a graph, this means that the graph $H^2$ can have an
arbitrarily large treewidth.

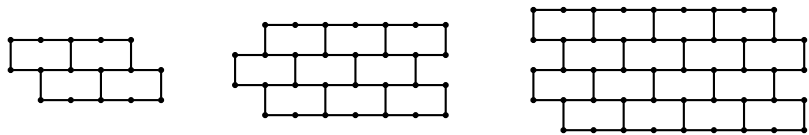
\begin{figure}
\begin{center}
\begin{minipage}{0.2\textwidth}
\centering
\begin{tikzpicture}[scale=0.4, every node/.style={scale=0.3}]
\GraphInit[vstyle=Simple]
\SetVertexSimple[MinSize=6pt]
\Vertex[x=1,y=0]{v10}
\Vertex[x=2,y=0]{v20}
\Vertex[x=3,y=0]{v30}
\Vertex[x=4,y=0]{v40}
\Vertex[x=5,y=0]{v50}

\Vertex[x=0,y=1]{v01}
\Vertex[x=1,y=1]{v11}
\Vertex[x=2,y=1]{v21}
\Vertex[x=3,y=1]{v31}
\Vertex[x=4,y=1]{v41}
\Vertex[x=5,y=1]{v51}

\Vertex[x=0,y=2]{v02}
\Vertex[x=1,y=2]{v12}
\Vertex[x=2,y=2]{v22}
\Vertex[x=3,y=2]{v32}
\Vertex[x=4,y=2]{v42}

\Edges(    v10,v20,v30,v40,v50)
\Edges(v01,v11,v21,v31,v41,v51)
\Edges(v02,v12,v22,v32,v42)

\Edge(v01)(v02)

\Edge(v10)(v11)

\Edge(v21)(v22)

\Edge(v30)(v31)

\Edge(v41)(v42)

\Edge(v50)(v51)

\end{tikzpicture}
\end{minipage}
\begin{minipage}{0.3\textwidth}
\centering
\begin{tikzpicture}[scale=0.4, every node/.style={scale=0.3}]
\GraphInit[vstyle=Simple]
\SetVertexSimple[MinSize=6pt]
\Vertex[x=1,y=0]{v10}
\Vertex[x=2,y=0]{v20}
\Vertex[x=3,y=0]{v30}
\Vertex[x=4,y=0]{v40}
\Vertex[x=5,y=0]{v50}
\Vertex[x=6,y=0]{v60}
\Vertex[x=7,y=0]{v70}

\Vertex[x=0,y=1]{v01}
\Vertex[x=1,y=1]{v11}
\Vertex[x=2,y=1]{v21}
\Vertex[x=3,y=1]{v31}
\Vertex[x=4,y=1]{v41}
\Vertex[x=5,y=1]{v51}
\Vertex[x=6,y=1]{v61}
\Vertex[x=7,y=1]{v71}

\Vertex[x=0,y=2]{v02}
\Vertex[x=1,y=2]{v12}
\Vertex[x=2,y=2]{v22}
\Vertex[x=3,y=2]{v32}
\Vertex[x=4,y=2]{v42}
\Vertex[x=5,y=2]{v52}
\Vertex[x=6,y=2]{v62}
\Vertex[x=7,y=2]{v72}

\Vertex[x=1,y=3]{v13}
\Vertex[x=2,y=3]{v23}
\Vertex[x=3,y=3]{v33}
\Vertex[x=4,y=3]{v43}
\Vertex[x=5,y=3]{v53}
\Vertex[x=6,y=3]{v63}
\Vertex[x=7,y=3]{v73}

\Edges(    v10,v20,v30,v40,v50,v60,v70)
\Edges(v01,v11,v21,v31,v41,v51,v61,v71)
\Edges(v02,v12,v22,v32,v42,v52,v62,v72)
\Edges(    v13,v23,v33,v43,v53,v63,v73)

\Edge(v01)(v02)

\Edge(v10)(v11)
\Edge(v12)(v13)

\Edge(v21)(v22)

\Edge(v30)(v31)
\Edge(v32)(v33)

\Edge(v41)(v42)

\Edge(v50)(v51)
\Edge(v52)(v53)

\Edge(v61)(v62)

\Edge(v70)(v71)
\Edge(v72)(v73)
\end{tikzpicture}
\end{minipage}
\begin{minipage}{0.35\textwidth}
\centering
\begin{tikzpicture}[scale=0.4, every node/.style={scale=0.3}]
\GraphInit[vstyle=Simple]
\SetVertexSimple[MinSize=6pt]
\Vertex[x=1,y=0]{v10}
\Vertex[x=2,y=0]{v20}
\Vertex[x=3,y=0]{v30}
\Vertex[x=4,y=0]{v40}
\Vertex[x=5,y=0]{v50}
\Vertex[x=6,y=0]{v60}
\Vertex[x=7,y=0]{v70}
\Vertex[x=8,y=0]{v80}
\Vertex[x=9,y=0]{v90}

\Vertex[x=0,y=1]{v01}
\Vertex[x=1,y=1]{v11}
\Vertex[x=2,y=1]{v21}
\Vertex[x=3,y=1]{v31}
\Vertex[x=4,y=1]{v41}
\Vertex[x=5,y=1]{v51}
\Vertex[x=6,y=1]{v61}
\Vertex[x=7,y=1]{v71}
\Vertex[x=8,y=1]{v81}
\Vertex[x=9,y=1]{v91}

\Vertex[x=0,y=2]{v02}
\Vertex[x=1,y=2]{v12}
\Vertex[x=2,y=2]{v22}
\Vertex[x=3,y=2]{v32}
\Vertex[x=4,y=2]{v42}
\Vertex[x=5,y=2]{v52}
\Vertex[x=6,y=2]{v62}
\Vertex[x=7,y=2]{v72}
\Vertex[x=8,y=2]{v82}
\Vertex[x=9,y=2]{v92}

\Vertex[x=0,y=3]{v03}
\Vertex[x=1,y=3]{v13}
\Vertex[x=2,y=3]{v23}
\Vertex[x=3,y=3]{v33}
\Vertex[x=4,y=3]{v43}
\Vertex[x=5,y=3]{v53}
\Vertex[x=6,y=3]{v63}
\Vertex[x=7,y=3]{v73}
\Vertex[x=8,y=3]{v83}
\Vertex[x=9,y=3]{v93}

\Vertex[x=0,y=4]{v04}
\Vertex[x=1,y=4]{v14}
\Vertex[x=2,y=4]{v24}
\Vertex[x=3,y=4]{v34}
\Vertex[x=4,y=4]{v44}
\Vertex[x=5,y=4]{v54}
\Vertex[x=6,y=4]{v64}
\Vertex[x=7,y=4]{v74}
\Vertex[x=8,y=4]{v84}

\Edges(    v10,v20,v30,v40,v50,v60,v70,v80,v90)
\Edges(v01,v11,v21,v31,v41,v51,v61,v71,v81,v91)
\Edges(v02,v12,v22,v32,v42,v52,v62,v72,v82,v92)
\Edges(v03,v13,v23,v33,v43,v53,v63,v73,v83,v93)
\Edges(v04,v14,v24,v34,v44,v54,v64,v74,v84)

\Edge(v01)(v02)
\Edge(v03)(v04)

\Edge(v10)(v11)
\Edge(v12)(v13)

\Edge(v21)(v22)
\Edge(v23)(v24)

\Edge(v30)(v31)
\Edge(v32)(v33)

\Edge(v41)(v42)
\Edge(v43)(v44)

\Edge(v50)(v51)
\Edge(v52)(v53)

\Edge(v61)(v62)
\Edge(v63)(v64)

\Edge(v70)(v71)
\Edge(v72)(v73)

\Edge(v81)(v82)
\Edge(v83)(v84)

\Edge(v90)(v91)
\Edge(v92)(v93)
\end{tikzpicture}
\end{minipage}
\caption{Walls of height 2, 3, and 4, respectively.}\label{f-walls}
\end{center}
\end{figure}

\section{Graphs of Maximum Degree at Most~6}\label{s-6}

In this section we show that the {\sc Square Root} problem can be solved in $O(n^4)$ time for $n$-vertex graphs of maximum degree at most~6.
In order to do this we need to consider the aforementioned generalization of {\sc Square Roots}, which is defined as follows.

\begin{description}
\item [{\sc Square Root with Labels}] 
\item[Input:] a graph $G$ and two sets of edges $R,B\subseteq E_G$.
\item[Question:] is there a graph $H$ with $H^2=G$, 
$R\subseteq E_H$ and $B\cap E_H=\emptyset$?
\end{description}

\noindent
Note that {\sc Square Root}  is indeed a special case of {\sc Square Root with Labels}: choose $R=B=\emptyset$.

The main idea behind our proof is to reduce to graphs with a bounded number of vertices by using the reduction rule that we recently introduced in~\cite{GKPS16} (the proof in~\cite{CochefertCGKP13} used a different and less general reduction rule, namely, the so-called path reduction rule, which only ensured boundedness of treewidth). In order to explain our new reduction we need the following definition.
An edge $uv$ of a graph $G$ is said to be \emph{recognizable} if 
$N_G(u)\cap N_G(v)$ has a partition $(X,Y)$, where $X=\{x_1,\ldots,x_p\}$ and $Y=\{y_1,\ldots,y_q\}$, for some $p,q\geq 1$, such that 
the following conditions are satisfied:
\begin{itemize}
\item[a)] $X$ and $Y$ are disjoint cliques in $G$;
\item[b)] $x_iy_j\notin E_G$ for $i\in\{1,\ldots,p\}$ and $j\in\{1,\ldots,q\}$;
\item[c)] for any $w\in N_G(u)\setminus (X\cup Y\cup \{v\})$, $wy_j\notin E_G$ for  $j\in\{1,\ldots,q\}$;
\item[d)]  for any $w\in N_G(v)\setminus (X\cup Y\cup \{u\})$, $wx_i\notin E_G$ for  $i\in\{1,\ldots,p\}$;
\item[e)] for any $w\in N_G(u)\setminus (X\cup Y\cup \{v\})$, there is an $i\in\{1,\ldots,p\}$ such that 
$wx_i\in E_G$;
\item[f)] for any $w\in N_G(v)\setminus (X\cup Y\cup \{u\})$, there is a $j\in\{1,\ldots,q\}$
such that $wy_j\in E_G$.
\end{itemize}
This leads to an algorithmic and a structural lemma, which we both need.

\begin{lemma}[\cite{GKPS16}]\label{lem:preproc}
For an instance $(G,R,B)$ of {\sc Square Root with Labels} where $G$ has $n$ vertices and $m$ 
edges, it takes  $O(n^2m^2)$ time to either solve the problem or to obtain an instance $(G',R',B')$ that has no recognizable edges and that is a yes-instance if and only if $(G,R,B)$ is a yes-instance.
\end{lemma}

\begin{lemma}[\cite{GKPS16}]\label{lem:edge-one2}
Let $H$ be a square root of a graph with no recognizable edges. Then every non-pendant edge of $H$ lies on a cycle of length at
most~$6$.
\end{lemma}

We are now ready to prove the aforementioned structural result. Its proof relies on Lemma~\ref{lem:edge-one2}.

\begin{lemma}\label{lem:size-bound}
Let $G$ be a connected graph with $\Delta(G)\leq 6$ that has no recognizable edges. If $G$ has a square root,
then $G$ has at most $103$ vertices. 
\end{lemma}

\begin{proof}
Let $G$ be a connected graph with $\Delta(G)\leq 6$ that has no recognizable edges.
Assume that $G$ has a square root~$H$.
We select a vertex $u$ for which there exists a vertex $v$, such that $\dist_H(u,v)=\diam(H)=s$. We apply a breadth-first search  on  
$H$ starting at $u$ to obtain levels $L_0,\ldots,L_s$, where $L_i=\{w\in V_G\; |\; \dist_G(u,w)=i\}$ for $i=0,\ldots,s$.
Note that $L_0=\{u\}$ and that $L_0\cup \cdots \cup L_s$ is a partition of $V_H=V_G$.
We say that a vertex $y$ is a \emph{child} of a vertex $x$, and that $x$ is a \emph{parent} of $y$
if $xy\in E_G$, $x\in L_i$ and $y\in L_{i+1}$ for some $i\in\{0,\ldots,s-1\}$. It is worth mentioning  
that this parent-child relation differs from the relation defined by the corresponding BFS-tree. 
In particular, a vertex may have several parents. We also say that a vertex~$z$ is a \emph{grandchild} of a vertex~$x$ and that $x$ is a \emph{grandparent} of $z$ if there is 
a vertex~$y$ such that $x$ is a parent of $y$ and $y$ is a parent of $z$. A vertex $y$ is a \emph{descendant} of a vertex $x$ if $x\in L_i$, $y\in L_j$ for some $i$, $j$ with $i<j$ and there is an $(x,y)$-path of length $|j-i|$.

We now prove a sequence of claims.

\medskip
\noindent
{\bf Claim A. }{\it 
Let $i\ge 2$ and $x\in L_i$ such that $x$ has at least two grandchildren.
Then $x$ has a child that is the parent of all grandchildren of $x$, while no other child of $x$ is the parent of a grandchild of $x$.}

\begin{figure}[ht]
\centering\scalebox{0.95}{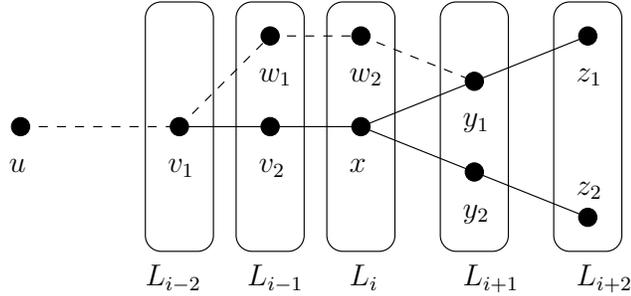}
\caption{The configuration for $x$ and its relatives in the graph $H$ of the proof of Claim~$A$.
\label{fig:conf}}
\end{figure}

\medskip
\noindent
We prove Claim~A as follows.
For contradiction, assume there exists a vertex $x\in L_i$ for some $i\geq 2$  that has two distinct children $y_1,y_2,$ such that $y_1$ and $y_2$ have distinct children $z_1,z_2$ respectively, as shown in Fig.~\ref{fig:conf}. Notice that possibly $z_1$  ($z_2$ respectively) is a child of $y_2$ ($y_1$ respectively) as well. As $i\geq 2$, $x$ has a parent $v_2$, and $v_2$ has a parent $v_1$.
By Lemma~\ref{lem:edge-one2}, $xv_2$  is contained in a cycle~$C$ of $H$ of length 
at most~6.
If $v_2$ is adjacent to a vertex $z\notin \{ v_1,x \}$ in $H$, then $z\in L_{i-2}\cup L_{i-1}\cup L_{i}$ and hence $z\notin \{ y_1,y_2,z_1,z_2 \}$. Consequently $d_G(x)\geq 7$; a contradiction. Therefore, $d_H(v_2)=2$ and thus $v_1v_2\in E_C$.
If $x$ has a neighbour $r\notin \{ v_2,y_1,y_2\}$ in $H$, then 
$d_G(x)\geq 7$, because $r\in L_{i-1}\cup L_{i}\cup L_{i+1}$ and hence $r$ is distinct from $v_1,z_1,z_2$. This is again a contradiction. 
Therefore $C$ contains one of $xy_1$, $xy_2$, say $C$ contains $xy_1$. 
Hence $C=v_1v_2xy_1w_2w_1v_1$ for some $w_1\in L_{i-1}$ and $w_2\in L_i$; see also Fig.~\ref{fig:conf}.
However, then $x$ is adjacent to $v_1, v_2, w_2, y_1, y_2, z_1, z_2$ in $G$. Hence $d_G(x)\geq 7$, a contradiction.  
This completes the proof of Claim~A.

\medskip
\noindent
{\bf Claim B. }{\it 
Let $i\geq 2$ and $x\in L_i$. Then the number of descendants of $x$ in $L_j$ is at most four for every $j>i$. }

\medskip
\noindent
We prove Claim B as follows.
Note that for all $i\ge 2$, every vertex $x\in L_i$  
has at least 
two neighbours in $G$ that belong to 
$L_{i-1}\cup L_{i-2}$. Hence the fact that
$d_G(x)\leq 6$ implies that $x$ has at most four neighbours in $G$
belonging to $L_{i+1}\cup L_{i+2}$.
In other words, the total number of children and grandchildren of $x$ is at most four.

We use induction on $i$.  
Let $i=s-1$ or $i=s-2$, 
As the total number of children and grandchildren 
of $x$ is at most four, the claim holds.
Let $i<s-2$. Recall that $x$ has at most four children. Hence the claim holds for $j=i+1$. Let $j>i+1$. If $x$ has no grandchildren the claim holds.
Now suppose that $x$ has exactly one grandchild $z$. If $j=i+2$, then the number of descendants of $x$ in $L_j$ is one. If $j>i+2$, then by the induction hypothesis the number of descendants of $x$ in $L_j$ is at most four, since these vertices are descendants of $z$ as well.
Finally suppose that $x$ has at least two grandchildren. 
Then, by Claim A, there is a unique child $y$ of $x$ that is the parent of all grandchildren of $x$,
and no other child of $x$ has children.  By the induction hypothesis, the number of descendants of $y$ in $L_j$ is at most four. As all descendants of $x$ in $L_j$  are descendants of $y$, the claim holds.
This completes the proof of Claim~B.

\medskip
\noindent
Recall that $v$ is a vertex that is of distance~$s$ of $x$.

\medskip
\noindent
{\bf Claim C. }{\it 
Let $P=x_0\cdots x_s$ with $x_0=u$ and $x_s=v$ 
be a shortest $(u,v)$-path in  $H$. 
Then for every $i\in\{3,\ldots,s-4\}$,  $x_{i+1}$ is the unique child of $x_i$.}

\medskip
\noindent
We prove Claim~C as follows.
To obtain a contradiction, we assume that $x_i$ has another child $y\neq x_{i+1}$ for some $i\in \{3\ldots s-3\}$, so $y\in L_{i+1}$.
Since $G$ has no recognizable edge, Lemma~\ref{lem:edge-one2} tells us that in $H$  every non-pendant edge is contained in a cycle of length at most $6$. 

Let us first assume that there is no cycle of length at most~6 in $H$ that contains both $x_{i-2}x_{i-1}$ and $x_{i-1}x_i$; this case will be studied later. 
Let $C_1$ and $C_2$ be two cycles in $H$ of length at most~ 6 that contain  $x_{i-2}x_{i-1}$ and $x_{i-1}x_i$, respectively.
As $x_{i-1}x_i\notin E_{C_1}$ and $x_{i-2}x_{i-1}\notin E_{C_2}$, $C_1$ has an edge $x_{i-1}w$ and  $C_2$ has an edge $x_{i-1}w'$ for some $w,w'\notin \{x_{i-2},x_i\}$. Note that $w$ and $w'$ both belong to $L_{i-2}\cup L_{i-1}\cup L_i$.
 If $w\neq w'$, then in $G$ we see that $x_{i-1}$ is adjacent to $x_{i-3}$, $x_{i-2}$, $x_i$, $x_{i+1}$, $w$, $w'$ and $y$. Hence
$d_G(x_{i-1})\geq 7$ 
contradicting $\Delta(G)\leq 6$. It follows that $w=w'$.
 
Let $z_1w$ be an edge of $C_1$ and let $z_2w$ be an edge of $C_2$, such that
$x_{i-1}\notin \{z_1,z_2\}$.
If $z_1$ or $z_2$ does not belong to $\{x_{i-3},x_{i-2},x_i,x_{i+1},y\}$, then $d_{G}(x_{i-1})\geq 7$, a contradiction. Hence, $\{z_1,z_2\}\subset 
\{x_{i-3},x_{i-2},x_i,x_{i+1},y\}$ 
Recall that there is no cycle of length at most~6 that contains both $x_{i-2}x_{i-1}$ and
$x_{i-1}x_i$. This implies that $z_1\notin \{x_i,x_{i+1},y\}$ (in $C_1$ remove $w$ in the first case and replace $w$ by $x_i$ in the latter two cases) and 
$z_2\notin \{x_{i-3},x_{i-2}\}$ (in $C_2$ replace $w$ by $x_{i-2}$ in the first case and remove $w$ in the second case). 
As both $z_1$ and $z_2$ are adjacent to $v$, this implies that $w\in L_{i-1}$, $z_1=x_{i-2}$ and $z_2=x_i$.
However, this means that $x_{i-2}x_{i-1}x_iwx_{i-2}$ is a cycle of length~$4$ in $H$ that contains both $x_{i-2}x_{i-1}$ and $x_{i-1}x_i$, a contradiction.  
Consequently, we may assume that there is a cycle $C$ in $H$ of length at most $6$ that contains both $x_{i-2}x_{i-1}$ and $x_{i-1}x_i$. 
Since $x_{i-2}\in L_{i-2}$, we need to distinguish three cases. 

\medskip
\noindent
{\bf Case 1}. $x_ix_{i+1}\in E_C$.\\ Then  $C$ has an edge $x_{i+1}w$ for some $w\in L_i$ such that $w\neq x_{i}$ (see Fig.~\ref{fig:case-1}).
Then, $C$ has an edge $wz$ for some $z\in L_{i-1}$. However, then we obtain $d_G(x_{i+1})\geq 7$, a contradiction.

\begin{figure}[ht]
\centering\scalebox{0.95}{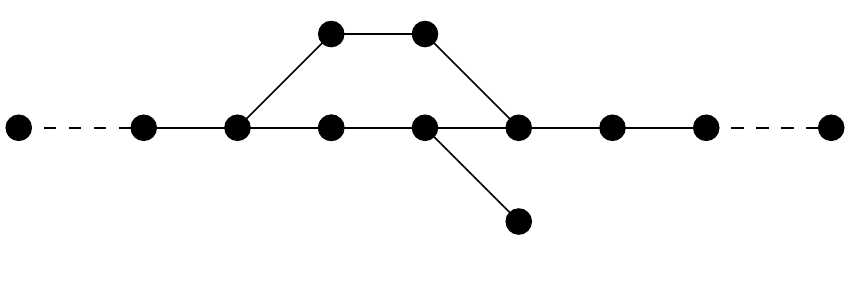}
\caption{Case 1.
\label{fig:case-1}}
\end{figure}

\medskip
\noindent
{\bf Case 2.} $x_iy\in E_C$.\\ Then  $C$ has an edge $yw$ for some $w\in L_i$ such that $w\neq x_{i}$ (see Fig.~\ref{fig:case-2}).
Similarly to Case~1, $C$ has an edge $wz$ for some $z\in L_{i-1}$. 
Note that $zx_{i-2}\in E_C$ as $C$ has length at most~6. Hence, $C=x_{i-2}x_{i-1}x_iywzx_{i-2}$
(note that this is not in contradiction with Claim~A as $i-2=1$ may hold).

If $x_iw\in E_H$, then $d_G(x_{i-1})\geq 7$, a contradiction. 
If $x_{i+1}w\in E_H$, then $d_G(x_{i+1})\geq 7$. If $x_{i+2}y\in E_H$, then $d_G(y)\geq 7$. If $x_{i-1}z\in E_H$ or $x_iz\in E_H$, then $d_G(x_i)\geq 7$.
Hence, $x_iw,x_{i+1}w,x_{i+2}y,x_{i-1}z,x_iz\notin E_H$.  

\begin{figure}[ht]
\centering\scalebox{0.95}{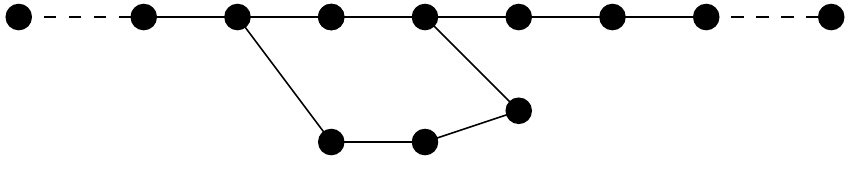}
\caption{Case 2.
\label{fig:case-2}}
\end{figure}

The edge $x_{i+1}x_{i+2}$ is non-pendant and thus included in a cycle $C'$ of length at most~6 by Lemma~\ref{lem:edge-one2}. Assume that $C'$ is a shortest cycle of this type. Let $x_{i+1}h$ be the other edge in $C'$ incident with $x_{i+1}$.
We distinguish three subcases.

\medskip
\noindent
{\bf Case 2a.} $h\neq x_i$.\\
If $h\neq y$, then $d_G(x_i)\geq 7$, a contradiction. Hence, $h=y$. Then, $C'$ has an edge $yg$ for some $g\neq x_{i+1}$. 
As $C'$ is a shortest cycle with $x_{i+1}x_{i+2}$, we find that $g\neq x_i$ (as otherwise the cycle obtained from $C'$ after removing $y$ would be shorter) . Recall that $x_{i+2}y\notin E_H$. Hence $g\neq x_{i+2}$.
If $g\neq w$, this means that $d_G(x_i)\geq 7$, a contradiction. Therefore, $g=w$. Let $f$ be the next vertex of $C'$, so $wf\in E_{C'}$. As 
$C'$ has length at most~6, $f\notin L_{i-1}$, so $f\notin \{x_{i-1},z\}$.
 Recall that $x_iw$ and $x_{i+1}w$ are not in $E_{H}$, thus $f\notin \{x_i,x_{i+1}\}$ either. As $w\in L_{i}$, $z\neq x_{i+2}$. 
 Hence, $d_G(y)\geq 7$, a contradiction.

\medskip
\noindent
{\bf Case 2b.} $h=x_i$ and $x_ig\in E_{C'}$ for  some $g\notin \{x_{i-1},x_{i+1}\}$.\\
Recall that $x_iw$ and $x_iz$ are not in $E_H$. Hence $g\notin \{w,z\}$ either. 
If $g\neq y$, this means that $d_G(x_i)\geq 7$. Hence, $g=y$, that is, $x_iy\in E_{C'}$. Let $f$ be the next vertex of $C'$, so $yf\in E_{C'}$. 
Recall that $x_{i+2}y\notin E_H$. If $f\neq w$, this means that $d_G(x_i)\geq 7$. It follows that $f=w$, that is, $yw\in E_{C'}$. Let $f'$ be the next vertex of $C'$, so $wf'\in E_{C'}$. As $C'$ has length at most~6, we find that $f'\in L_{i+1}$.
As $C'$ has length at most 6, we find that $f'x_{i+2}\in E_H$ and thus $C'=x_{i+2}x_{i+1}x_iywf'x_{i+2}$ (see  Fig.~\ref{fig:case-2b}). 

If $x_{i+2}$ has a neighbour in $H$ distinct from $x_{i+1},x_{i+3},f'$, then $d_G(x_{i+2})\geq 7$, a contradiction. Hence, $N_H(x_{i+2})=\{x_{i+1},x_{i+3},f'\}$. If $x_{i+1}$ has a neighbour in $H$ distinct from $x_i$ and $x_{i+2}$, then that neighbour cannot be in $\{y,w,f'\}$, as $C'$ is a shortest cycle containing $x_{i+1}x_{i+2}$.
Consequently, we find that $d_G(x_{i+1})\geq 7$, a contradiction. Hence, $N_H(x_{i+1})=\{x_i,x_{i+2}\}$.
Recall that $x_iz\notin E_H$. Moreover, as $C'$ is a shortest cycle containing $x_{i+1}x_{i+2}$, we find that $x_iw, x_if'\notin E_H$.
Consequently, if $x_i$ has a neighbour in $H$ distinct from $x_{i-1},x_{i+1},y$, then $d_G(y)\geq 7$, a contradiction. Hence, $N_H(x_i)=\{x_{i-1},x_{i+1},y\}$. 
As $C'$  is a shortest cycle containing $x_{i+1}x_{i+2}$, we find that $yf' \notin E_H$.
Consequently, if $y$ has a neighbour in $H$ distinct from $x_i,w$, then $d_G(y)\geq 7$, a contradiction. Hence, $N_H(y)=\{x_i,w\}$. 
If $w$ has a neighbour distinct from $y,z,f'$, then $d_G(y)\geq 7$, a contradiction. Hence, $N_H(w)=\{y,z,f'\}$.
If $f'$ has a neighbour distinct from $w,x_{i+2}$, then $d_G(w)\geq 7$, a contradiction. Hence, $N_H(f')=\{x_{i+2},w\}$.
  
Now consider the (non-pendant) edge $x_{i+2}x_{i+3}$, which must be in a cycle~$C''$ of length at most~6 due to Lemma~\ref{lem:edge-one2}. 
As $x_{i+3}\in L_{i+3}$ and $|E_{C''}|\leq 6$, we find that $C''$ contains no vertex of $L_{i-1}$. Now, by traversing $C''$ starting at $x_{i+2}$ in opposite direction from $x_{i+3}$, we find that $C''$ contains the cycle $x_{i+2}x_{i+1}x_iywf'x_{i+2}$. Hence $d_{C''}(x_{i+2})\geq 3$. which means that $C''$ is not a cycle, a contradiction.

\begin{figure}[ht]
\centering\scalebox{0.95}{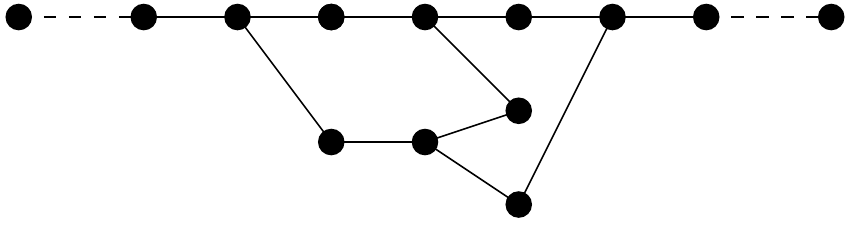}
\caption{Case 2b.
\label{fig:case-2b}}
\end{figure}

\medskip
\noindent
{\bf Case 2c.} $h=x_i$ and $x_ix_{i-1}\in E_{C'}$\\ 
Let $g$ be the next vertex of $C'$, so $x_{i-1}g\in E_{C'}$.
As $C'$ has length at most~6, we find that $g\in L_i$. Recall that $x_{i-1}w\notin E_H$. Then we find that $d_G(x_i)\geq 7$, a contradiction.

\medskip
\noindent
{\bf Case 3.}  $x_ix_{i+1},x_iy\notin E_C$.\\
Then $C$ has an edge $x_iw$ for some $w\notin \{x_{i-1},x_{i+1},y\}$ (see Fig.~\ref{fig:case-3}). 
Let $z$ be the next vertex of $C'$, so $wz\in E_C$. 
As $C$ has length at most~6, we find that $z\notin L_{i+1}\cup L_{i+2}$.
Hence, if $z\neq x_{i-2}$, then $d_G(x_i)\geq 7$, a contradiction. This means that $z=x_{i-2}$, and consequently $C=x_{i-2}x_{i-1}x_iw,x_{i-2}$ and $w\in L_{i-1}$. 

Again let $C'$ be a shortest cycle amongst all cycles that contains $x_{i+1}x_{i+2}$. Recall that the length of $C'$ is at most~6 by Lemma~\ref{lem:edge-one2}. Let $x_{i+1}h$ be the other edge in $C'$ incident with $x_{i+1}$. As in Case 2, we distinguish three subcases.

\begin{figure}[ht]
\centering\scalebox{0.95}{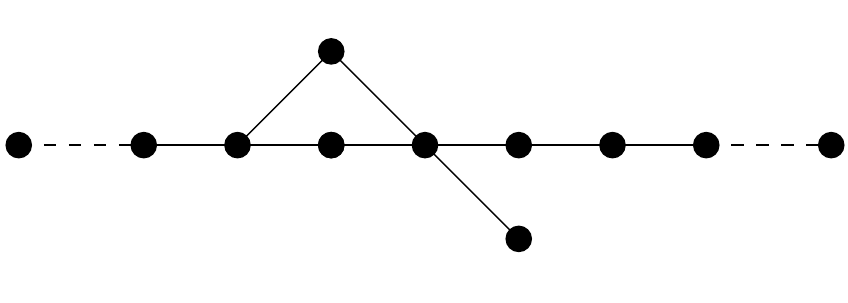}
\caption{Case 3.
\label{fig:case-3}}
\end{figure}

\medskip
\noindent
{\bf Case 3a.} $h\neq x_i$.\\
If $h\neq y$, then $d_G(x_i)\geq 7$. Hence, $h=y$, that is, $x_{i+1}y\in E_{C'}$. 
Let $g$ be the next vertex of $C'$, so $yg\in E_{C'}$. 
As $C'$ is a shortest cycle containing $x_{i+1}x_{i+2}$, we find that $g\neq x_i$ (otherwise remove $y$ from $C'$ to obtain a shorter cycle).
As $yx_{i+2}\notin E_H$ due to Claim~A, we find that $g\neq x_{i+2}$. 
Since $y\in L_{i+1}$ and $w\in L_{i-1}$, we find that $g\neq w$ either. 
Then $d_G(x_i)\geq 7$ holds, a contradiction.

\medskip
\noindent
{\bf Case 3b.} $h=x_i$ and $x_ig\in E_{C'}$ for  some $g\notin \{x_{i-1},x_{i+1}\}$.\\
If $g\notin \{w,y\}$, then $d_G(x_i)\geq 7$. 
First suppose $g=w$, so
$x_iw\in E_{C'}$. Let $f$ be the next vertex of $C'$, so $wf\in E_{C'}$.
As $w\in L_{i-1}$ and $C'$ has length at most~6, we find that $f\in L_i$.
However, then $d_G(x_i)\geq 7$, a contradiction.
Now suppose, $g=y$, so $x_iy\in E_{C'}$. Let $f'$ be the next vertex of $C'$, so $yf'\in E_{C'}$.  
Recall that $yx_{i+2}\notin E_H$ due to Claim~A.
Then $d_G(x_i)\geq 7$ holds, a contradiction.

\medskip
\noindent
{\bf Case 3c.} $h=x_i$ and $x_ix_{i-1}\in E_{C'}$\\
As $C'$ has length at most~6, we find that $h\in L_i$. Then $d_G(x)\geq 7$, a contradiction.

\medskip
\noindent
We considered all possibilities, and in each case we obtained a contradiction. Hence we have proven Claim~C.

\medskip
\noindent
{\bf Claim D. }{\it 
Let $P=x_0\cdots x_s$ with $x_0=u$ and $x_s=v$ 
be a shortest $(u,v)$-path in $H$.
Then for every $i\in\{4,\ldots,s-3\}$, $x_{i-1}$ is the unique parent of $x_i$.}

\medskip
\noindent
To prove Claim~D, it suffices to run a breadth-first search in $H$ from $v$, to consider the resulting levels $L_0',\ldots,L_{s}'$,
where where $L_i'=\{w\in V_G\; |\; \dist_G(v,w)=i\}$ for $i=0,\ldots,s$, and to apply Claim~C.

\medskip
\noindent
We are now ready to complete the proof. We do this by first showing  that $s\leq 8$. To obtain a contradiction, assume that $s\geq 9$. 
By Lemma~\ref{lem:edge-one2}, $H$ has a cycle $C$ of length at most 6 that contains $x_5x_6$. As $x_6$ is the unique child of $x_5$ by Claim~C and $x_5$ is the unique parent of $x_6$ by Claim~D, we find that $C$ has an edge $yz$, where $y\in L_5$, $z\in L_6$, such that $C$ contains an $(x_5,y)$ path $Q$ of length at most 3 with $V_Q\subseteq L_4\cup L_5$.  

Suppose that $x_5y\in E_Q$. Then $y$ has a parent $h\in L_4$. 
By Claim~C, $h\neq x_4$. Therefore, $d_G(x_5)\geq 7$, a contradiction. Hence $x_5y\notin E_Q$. 
Suppose that $V_Q\subseteq L_5$. Then $Q$ has edges $x_5w,wh$ such that 
$w,h\in L_5$. Again, $w$ has a parent $g\in L_4$ and $g\neq x_4$ due to Claim~C. It follows that $d_G(x_5)\geq 7$; a contradiction. Hence, $Q$ has a vertex of $L_4$.

As $Q$ contains a vertex of $L_4$, $Q$ has length at least~2. If $Q$ has length~2, then $Q=x_5x_4y$. However, this is a contradiction, as
$x_5$ is the unique child of $x_4$ by Claim~C. Hence, $Q$ has length~3, which implies that $Q=x_5x_4zy$ for some $z\in L_4$ or
$Q=x_5ww'y$ for some $w\in L_5$ and $w'\in L_4$.   

First suppose that $Q=x_5x_4zy$ for some $h\in L_4$. Let $g\in L_3$ be a parent of $h$. As $x_4$ is the unique child of $x_3$ by Claim~C,
we find that $g\neq x_4$. Then $d_G(x_4)\geq 7$, a contradiction.
Now suppose that $Q=x_5ww'y$ for some $w\in L_5$ and $w'\in L_4$. Note that $w'\neq x_4$, as $x_5$ is the unique child of $x_4$ by Claim $C$. 
As $C$ has length at most~6, we find that $C=x_6x_5ww'yzx_6$. This means that in $G$, $x_5$ is adjacent to $x_3,x_4,x_6,x_7,w,w',z$, so
$d_G(x_5)\geq 7$, a contradiction. We conclude that $s\leq 8$.

As $L_0=\{u\}$, we find that $|L_0|=1$. 
Since $d_G(u)\leq 6$, $|L_1\cup L_2|\leq 6$. As $|L_1|\geq 1$, this means that $|L_2|\leq 5$. 
If $|L_2|=5$, then $L_1=\{v\}$ for some $v\in V_H$ and each vertex of $L_2$ is a child of $v$. As $d_G(v)\leq 5$, 
this means that $V_G=V_H=L_0\cup L_1\cup L_2$, so $|V_H|=1+1+5=7$. Suppose $|L_2|\leq 4$. By Claim B, $|L_i|\leq 4|L_2|\leq 16$ for $i\geq 3$. 
Because $s\leq 8$, $|V_G|=|V_H|=|L_0|+|L_1\cup L_2|+ |L_3| + \cdots +|L_8|\leq 1 + 6 + 6\cdot 16 = 103$.\qed
\end{proof}

We are now ready to prove our main result. Its proof uses Lemmas~\ref{lem:preproc} and~\ref{lem:size-bound}.

\begin{theorem}\label{thm:root}
 {\sc Square Root} can be solved in time $O(n^4)$ for  $n$-vertex graphs of maximum degree at most~$6$.  
\end{theorem}

\begin{proof}
Let $G$ be a graph of maximum degree at most 6. 
We construct an instance $(G,R,B)$ of  {\sc Square Root with Labels}  from 
$G$ by setting $R=B=\emptyset$. Then we preprocess $(G,R,B)$ using Lemma~\ref{lem:preproc}.
In this way we either solve the problem (and answer {\tt no}) or obtain an equivalent instance $(G',R',B')$ of  {\sc Square Root with Labels} that has no recognizable edges. 
In the latter case we know from Lemma~\ref{lem:size-bound} that if $G'$ has a square root, then each component of $G'$ 
has at most 103 vertices. Hence, if $G'$ has a component with at least 100 vertices, then we stop and return {\tt no}. Otherwise, we solve the problem for each component of $G'$ in constant time by brute force.
Applying Lemma~\ref{lem:preproc} takes time $O(n^2m^2)=O(n^4)$, as $\Delta(G)\leq 6$. 
The remainder of our algorithm takes constant time. Hence, its total running time is $O(n^4)$.\qed
\end{proof}

\medskip
\noindent
{\bf Remark~2.}
We cannot extend this approach for graphs of maximum degree at most~7. In particular, 
the following example shows that we cannot obtain an analog of Lemma~\ref{lem:size-bound} for graphs of maximum degree at most 7. Let $H$ be the graph obtained from two paths $u_1,\ldots,u_n$ and $v_1,\ldots,v_n$ by adding the edge~$u_iv_i$ for $i\in\{1,\ldots,n\}$. Notice that $G=H^2$ has no recognizable edges while
$\Delta(G)=7$. However, its size $|V_G|=|V_H|=2n$  can be arbitrarily large.

\section{Conclusion}\label{s-con}

We showed that {\sc Square Root} can be solved in polynomial time for graphs of maximum degree at most~6. We pose two open problems. First, can we solve {\sc Square Root} in polynomial time on graphs of maximum degree at most~7; this will require new techniques.
Second, does there exist an integer~$k$ such that {\sc Square Root} is \NP-complete for graphs of maximum degree at most~$k$?

\end{document}

%% file: Fig8.pdf_t
\begin{picture}(0,0)%
\includegraphics{Fig8.pdf}%
\end{picture}%
\setlength{\unitlength}{3947sp}%
\begingroup\makeatletter\ifx\SetFigFont\undefined%
\gdef\SetFigFont#1#2#3#4#5{%
  \reset@font\fontsize{#1}{#2pt}%
  \fontfamily{#3}\fontseries{#4}\fontshape{#5}%
  \selectfont}%
\fi\endgroup%
\begin{picture}(5836,1340)(383,-584)
\put(5176,-511){\makebox(0,0)[lb]{\smash{{\SetFigFont{12}{14.4}{\rmdefault}{\mddefault}{\updefault}{\color[rgb]{0,0,0}c)}%
}}}}
\put(1276,-61){\makebox(0,0)[lb]{\smash{{\SetFigFont{12}{14.4}{\rmdefault}{\mddefault}{\updefault}{\color[rgb]{0,0,0}$x$}%
}}}}
\put(3076,-61){\makebox(0,0)[lb]{\smash{{\SetFigFont{12}{14.4}{\rmdefault}{\mddefault}{\updefault}{\color[rgb]{0,0,0}$v$}%
}}}}
\put(5176,-211){\makebox(0,0)[lb]{\smash{{\SetFigFont{12}{14.4}{\rmdefault}{\mddefault}{\updefault}{\color[rgb]{0,0,0}$v$}%
}}}}
\put(1276,-511){\makebox(0,0)[lb]{\smash{{\SetFigFont{12}{14.4}{\rmdefault}{\mddefault}{\updefault}{\color[rgb]{0,0,0}a)}%
}}}}
\put(3076,-511){\makebox(0,0)[lb]{\smash{{\SetFigFont{12}{14.4}{\rmdefault}{\mddefault}{\updefault}{\color[rgb]{0,0,0}b)}%
}}}}
\end{picture}%

%% file: Fig2.pdf_t
\begin{picture}(0,0)%
\includegraphics{Fig2.pdf}%
\end{picture}%
\setlength{\unitlength}{3947sp}%
\begingroup\makeatletter\ifx\SetFigFont\undefined%
\gdef\SetFigFont#1#2#3#4#5{%
  \reset@font\fontsize{#1}{#2pt}%
  \fontfamily{#3}\fontseries{#4}\fontshape{#5}%
  \selectfont}%
\fi\endgroup%
\begin{picture}(4077,1965)(511,-1789)
\put(2776,-361){\makebox(0,0)[lb]{\smash{{\SetFigFont{12}{14.4}{\rmdefault}{\mddefault}{\updefault}$w_2$}}}}
\put(2776,-961){\makebox(0,0)[lb]{\smash{{\SetFigFont{12}{14.4}{\rmdefault}{\mddefault}{\updefault}$x$}}}}
\put(3526,-661){\makebox(0,0)[lb]{\smash{{\SetFigFont{12}{14.4}{\rmdefault}{\mddefault}{\updefault}$y_1$}}}}
\put(3526,-1261){\makebox(0,0)[lb]{\smash{{\SetFigFont{12}{14.4}{\rmdefault}{\mddefault}{\updefault}$y_2$}}}}
\put(4276,-361){\makebox(0,0)[lb]{\smash{{\SetFigFont{12}{14.4}{\rmdefault}{\mddefault}{\updefault}$z_1$}}}}
\put(4276,-1111){\makebox(0,0)[lb]{\smash{{\SetFigFont{12}{14.4}{\rmdefault}{\mddefault}{\updefault}$z_2$}}}}
\put(1576,-961){\makebox(0,0)[lb]{\smash{{\SetFigFont{12}{14.4}{\rmdefault}{\mddefault}{\updefault}$v_1$}}}}
\put(2176,-961){\makebox(0,0)[lb]{\smash{{\SetFigFont{12}{14.4}{\rmdefault}{\mddefault}{\updefault}$v_2$}}}}
\put(526,-961){\makebox(0,0)[lb]{\smash{{\SetFigFont{12}{14.4}{\rmdefault}{\mddefault}{\updefault}$u$}}}}
\put(2776,-1711){\makebox(0,0)[lb]{\smash{{\SetFigFont{12}{14.4}{\rmdefault}{\mddefault}{\updefault}$L_i$}}}}
\put(3526,-1711){\makebox(0,0)[lb]{\smash{{\SetFigFont{12}{14.4}{\rmdefault}{\mddefault}{\updefault}$L_{i+1}$}}}}
\put(4276,-1711){\makebox(0,0)[lb]{\smash{{\SetFigFont{12}{14.4}{\rmdefault}{\mddefault}{\updefault}$L_{i+2}$}}}}
\put(1426,-1711){\makebox(0,0)[lb]{\smash{{\SetFigFont{12}{14.4}{\rmdefault}{\mddefault}{\updefault}$L_{i-2}$}}}}
\put(2101,-1711){\makebox(0,0)[lb]{\smash{{\SetFigFont{12}{14.4}{\rmdefault}{\mddefault}{\updefault}$L_{i-1}$}}}}
\put(2176,-361){\makebox(0,0)[lb]{\smash{{\SetFigFont{12}{14.4}{\rmdefault}{\mddefault}{\updefault}$w_1$}}}}
\end{picture}%

%% file: Fig3.pdf_t
\begin{picture}(0,0)%
\includegraphics{Fig3.pdf}%
\end{picture}%
\setlength{\unitlength}{3947sp}%
\begingroup\makeatletter\ifx\SetFigFont\undefined%
\gdef\SetFigFont#1#2#3#4#5{%
  \reset@font\fontsize{#1}{#2pt}%
  \fontfamily{#3}\fontseries{#4}\fontshape{#5}%
  \selectfont}%
\fi\endgroup%
\begin{picture}(4058,1362)(211,-814)
\put(2776,-736){\makebox(0,0)[lb]{\smash{{\SetFigFont{12}{14.4}{\rmdefault}{\mddefault}{\updefault}{\color[rgb]{0,0,0}$y$}%
}}}}
\put(2101,-361){\makebox(0,0)[lb]{\smash{{\SetFigFont{12}{14.4}{\rmdefault}{\mddefault}{\updefault}{\color[rgb]{0,0,0}$x_i$}%
}}}}
\put(2626,-361){\makebox(0,0)[lb]{\smash{{\SetFigFont{12}{14.4}{\rmdefault}{\mddefault}{\updefault}{\color[rgb]{0,0,0}$x_{i+1}$}%
}}}}
\put(3076,-361){\makebox(0,0)[lb]{\smash{{\SetFigFont{12}{14.4}{\rmdefault}{\mddefault}{\updefault}{\color[rgb]{0,0,0}$x_{i+2}$}%
}}}}
\put(1726,-361){\makebox(0,0)[lb]{\smash{{\SetFigFont{12}{14.4}{\rmdefault}{\mddefault}{\updefault}{\color[rgb]{0,0,0}$x_{i-1}$}%
}}}}
\put(1276,-361){\makebox(0,0)[lb]{\smash{{\SetFigFont{12}{14.4}{\rmdefault}{\mddefault}{\updefault}{\color[rgb]{0,0,0}$x_{i-2}$}%
}}}}
\put(4126,-361){\makebox(0,0)[lb]{\smash{{\SetFigFont{12}{14.4}{\rmdefault}{\mddefault}{\updefault}{\color[rgb]{0,0,0}$v=x_{s(u)}$}%
}}}}
\put(226,-361){\makebox(0,0)[lb]{\smash{{\SetFigFont{12}{14.4}{\rmdefault}{\mddefault}{\updefault}{\color[rgb]{0,0,0}$u=x_0$}%
}}}}
\put(2401,389){\makebox(0,0)[lb]{\smash{{\SetFigFont{12}{14.4}{\rmdefault}{\mddefault}{\updefault}{\color[rgb]{0,0,0}$w$}%
}}}}
\put(1576,389){\makebox(0,0)[lb]{\smash{{\SetFigFont{12}{14.4}{\rmdefault}{\mddefault}{\updefault}{\color[rgb]{0,0,0}$z$}%
}}}}
\end{picture}%

%% file: Fig4.pdf_t
\begin{picture}(0,0)%
\includegraphics{Fig4.pdf}%
\end{picture}%
\setlength{\unitlength}{3947sp}%
\begingroup\makeatletter\ifx\SetFigFont\undefined%
\gdef\SetFigFont#1#2#3#4#5{%
  \reset@font\fontsize{#1}{#2pt}%
  \fontfamily{#3}\fontseries{#4}\fontshape{#5}%
  \selectfont}%
\fi\endgroup%
\begin{picture}(4058,886)(211,-880)
\put(2401,-811){\makebox(0,0)[lb]{\smash{{\SetFigFont{12}{14.4}{\rmdefault}{\mddefault}{\updefault}{\color[rgb]{0,0,0}$w$}%
}}}}
\put(2101,-361){\makebox(0,0)[lb]{\smash{{\SetFigFont{12}{14.4}{\rmdefault}{\mddefault}{\updefault}{\color[rgb]{0,0,0}$x_i$}%
}}}}
\put(2626,-361){\makebox(0,0)[lb]{\smash{{\SetFigFont{12}{14.4}{\rmdefault}{\mddefault}{\updefault}{\color[rgb]{0,0,0}$x_{i+1}$}%
}}}}
\put(3076,-361){\makebox(0,0)[lb]{\smash{{\SetFigFont{12}{14.4}{\rmdefault}{\mddefault}{\updefault}{\color[rgb]{0,0,0}$x_{i+2}$}%
}}}}
\put(1726,-361){\makebox(0,0)[lb]{\smash{{\SetFigFont{12}{14.4}{\rmdefault}{\mddefault}{\updefault}{\color[rgb]{0,0,0}$x_{i-1}$}%
}}}}
\put(1276,-361){\makebox(0,0)[lb]{\smash{{\SetFigFont{12}{14.4}{\rmdefault}{\mddefault}{\updefault}{\color[rgb]{0,0,0}$x_{i-2}$}%
}}}}
\put(4126,-361){\makebox(0,0)[lb]{\smash{{\SetFigFont{12}{14.4}{\rmdefault}{\mddefault}{\updefault}{\color[rgb]{0,0,0}$v=x_{s(u)}$}%
}}}}
\put(226,-361){\makebox(0,0)[lb]{\smash{{\SetFigFont{12}{14.4}{\rmdefault}{\mddefault}{\updefault}{\color[rgb]{0,0,0}$u=x_0$}%
}}}}
\put(2776,-736){\makebox(0,0)[lb]{\smash{{\SetFigFont{12}{14.4}{\rmdefault}{\mddefault}{\updefault}{\color[rgb]{0,0,0}$y$}%
}}}}
\put(1576,-811){\makebox(0,0)[lb]{\smash{{\SetFigFont{12}{14.4}{\rmdefault}{\mddefault}{\updefault}{\color[rgb]{0,0,0}$z$}%
}}}}
\end{picture}%

%% file: Fig4b.pdf_t
\begin{picture}(0,0)%
\includegraphics{Fig4b.pdf}%
\end{picture}%
\setlength{\unitlength}{3947sp}%
\begingroup\makeatletter\ifx\SetFigFont\undefined%
\gdef\SetFigFont#1#2#3#4#5{%
  \reset@font\fontsize{#1}{#2pt}%
  \fontfamily{#3}\fontseries{#4}\fontshape{#5}%
  \selectfont}%
\fi\endgroup%
\begin{picture}(4058,1186)(211,-1180)
\put(2651,-767){\makebox(0,0)[lb]{\smash{{\SetFigFont{12}{14.4}{\rmdefault}{\mddefault}{\updefault}{\color[rgb]{0,0,0}$y$}%
}}}}
\put(2101,-361){\makebox(0,0)[lb]{\smash{{\SetFigFont{12}{14.4}{\rmdefault}{\mddefault}{\updefault}{\color[rgb]{0,0,0}$x_i$}%
}}}}
\put(2626,-361){\makebox(0,0)[lb]{\smash{{\SetFigFont{12}{14.4}{\rmdefault}{\mddefault}{\updefault}{\color[rgb]{0,0,0}$x_{i+1}$}%
}}}}
\put(1726,-361){\makebox(0,0)[lb]{\smash{{\SetFigFont{12}{14.4}{\rmdefault}{\mddefault}{\updefault}{\color[rgb]{0,0,0}$x_{i-1}$}%
}}}}
\put(1276,-361){\makebox(0,0)[lb]{\smash{{\SetFigFont{12}{14.4}{\rmdefault}{\mddefault}{\updefault}{\color[rgb]{0,0,0}$x_{i-2}$}%
}}}}
\put(226,-361){\makebox(0,0)[lb]{\smash{{\SetFigFont{12}{14.4}{\rmdefault}{\mddefault}{\updefault}{\color[rgb]{0,0,0}$u=x_0$}%
}}}}
\put(1576,-811){\makebox(0,0)[lb]{\smash{{\SetFigFont{12}{14.4}{\rmdefault}{\mddefault}{\updefault}{\color[rgb]{0,0,0}$z$}%
}}}}
\put(2851,-1111){\makebox(0,0)[lb]{\smash{{\SetFigFont{12}{14.4}{\rmdefault}{\mddefault}{\updefault}{\color[rgb]{0,0,0}$f'$}%
}}}}
\put(4126,-361){\makebox(0,0)[lb]{\smash{{\SetFigFont{12}{14.4}{\rmdefault}{\mddefault}{\updefault}{\color[rgb]{0,0,0}$v=x_{s(u)}$}%
}}}}
\put(3076,-361){\makebox(0,0)[lb]{\smash{{\SetFigFont{12}{14.4}{\rmdefault}{\mddefault}{\updefault}{\color[rgb]{0,0,0}$x_{i+2}$}%
}}}}
\put(3526,-361){\makebox(0,0)[lb]{\smash{{\SetFigFont{12}{14.4}{\rmdefault}{\mddefault}{\updefault}{\color[rgb]{0,0,0}$x_{i+3}$}%
}}}}
\put(2428,-767){\makebox(0,0)[lb]{\smash{{\SetFigFont{12}{14.4}{\rmdefault}{\mddefault}{\updefault}{\color[rgb]{0,0,0}$w$}%
}}}}
\end{picture}%

%% file: Fig5.pdf_t
\begin{picture}(0,0)%
\includegraphics{Fig5.pdf}%
\end{picture}%
\setlength{\unitlength}{3947sp}%
\begingroup\makeatletter\ifx\SetFigFont\undefined%
\gdef\SetFigFont#1#2#3#4#5{%
  \reset@font\fontsize{#1}{#2pt}%
  \fontfamily{#3}\fontseries{#4}\fontshape{#5}%
  \selectfont}%
\fi\endgroup%
\begin{picture}(4058,1437)(211,-814)
\put(1876,464){\makebox(0,0)[lb]{\smash{{\SetFigFont{12}{14.4}{\rmdefault}{\mddefault}{\updefault}{\color[rgb]{0,0,0}$w$}%
}}}}
\put(2101,-361){\makebox(0,0)[lb]{\smash{{\SetFigFont{12}{14.4}{\rmdefault}{\mddefault}{\updefault}{\color[rgb]{0,0,0}$x_i$}%
}}}}
\put(2626,-361){\makebox(0,0)[lb]{\smash{{\SetFigFont{12}{14.4}{\rmdefault}{\mddefault}{\updefault}{\color[rgb]{0,0,0}$x_{i+1}$}%
}}}}
\put(3076,-361){\makebox(0,0)[lb]{\smash{{\SetFigFont{12}{14.4}{\rmdefault}{\mddefault}{\updefault}{\color[rgb]{0,0,0}$x_{i+2}$}%
}}}}
\put(1726,-361){\makebox(0,0)[lb]{\smash{{\SetFigFont{12}{14.4}{\rmdefault}{\mddefault}{\updefault}{\color[rgb]{0,0,0}$x_{i-1}$}%
}}}}
\put(1276,-361){\makebox(0,0)[lb]{\smash{{\SetFigFont{12}{14.4}{\rmdefault}{\mddefault}{\updefault}{\color[rgb]{0,0,0}$x_{i-2}$}%
}}}}
\put(4126,-361){\makebox(0,0)[lb]{\smash{{\SetFigFont{12}{14.4}{\rmdefault}{\mddefault}{\updefault}{\color[rgb]{0,0,0}$v=x_{s(u)}$}%
}}}}
\put(226,-361){\makebox(0,0)[lb]{\smash{{\SetFigFont{12}{14.4}{\rmdefault}{\mddefault}{\updefault}{\color[rgb]{0,0,0}$u=x_0$}%
}}}}
\put(2776,-736){\makebox(0,0)[lb]{\smash{{\SetFigFont{12}{14.4}{\rmdefault}{\mddefault}{\updefault}{\color[rgb]{0,0,0}$y$}%
}}}}
\end{picture}%